\documentclass[twocolumn,letterpaper]{article}

\usepackage{graphicx}           
\usepackage{amssymb}            
\usepackage{amsthm}
\usepackage{unicode-math}       

\usepackage[svgnames]{xcolor}
\usepackage[colorlinks=true,linkcolor=RoyalBlue,citecolor=Tomato,urlcolor=DarkOrchid]{hyperref}

\usepackage{printlen}
\uselengthunit{in}

\usepackage{authblk}

\makeatletter
\let\c@author\relax
\makeatother

\newtheorem{theorem}{Theorem}

\newtheorem{corollary}{Corollary}

\newtheorem*{remark}{Remark}

\DeclareMathOperator{\Cov}{Cov}
\DeclareMathOperator{\Corr}{Corr}

\usepackage[maxnames=1,backend=biber,sorting=none,giveninits=true,style=numeric-comp,hyperref=true,uniquename=init,date=year]{biblatex}
\DeclareSourcemap{
  \maps[datatype=bibtex]{
    \map{
      \step[fieldsource=doi,final]
      \step[fieldset=url,null]
    }  
  }
}

\usepackage{flushend}

\usepackage{cleveref}
\crefname{figure}{Figure}{Figures}
\crefname{table}{Table}{Tables}
\crefname{equation}{Equation}{Equations}
\crefname{theorem}{Theorem}{Theorems}
\crefname{corollary}{Corollary}{Corollaries}
\crefname{lemma}{Lemma}{Lemmas}

\AtEveryBibitem{\ifentrytype{article}{\clearfield{url}\clearfield{urlyear}}{}}

\DeclareNameAlias{sortname}{family-given}

\title{An intuitive rearranging of the Yates covariance decomposition for probabilistic verification of forecasts with the Brier score}

\author{Bruno Hebling Vieira}
\affil{Methods of Plasticity Research, Department of Psychology, University of Zurich, Zurich, Switzerland \\ Orcid: \href{https://orcid.org/0000-0002-8770-7396}{0000-0002-8770-7396} \\ \href{mailto:bruno.heblingvieira@uzh.ch}{bruno.heblingvieira@uzh.ch}}

\date{}

\begin{document}

\maketitle

\begin{abstract}
  Proper scoring rules are essential for evaluating probabilistic forecasts.
  We propose a simple algebraic rearrangement of the Yates covariance decomposition of the Brier score into three independently non-negative terms: a variance mismatch term, a correlation deficit term, and a calibration-in-the-large term.
  This rearrangement makes the optimality conditions for perfect forecasting transparent: the optimal forecast must simultaneously match the variance of outcomes, achieve perfect positive correlation with outcomes, and match the mean of outcomes.
  Any deviation from these conditions results in a positive contribution to the Brier score.
\end{abstract}

\section{Introduction}

Proper scoring rules are a fundamental concept in the evaluation of probabilistic forecasts~\cite{Gneiting_JAmStatAssoc_2007}.
They provide a principled way to assess the accuracy of predicted probabilities by assigning a numerical score based on the predicted probability and the actual outcome.
A scoring rule is considered proper if it encourages truthful reporting of probabilities, meaning that the expected score is maximized when the predicted probabilities match the true probabilities.

One of the most widely used proper scoring rules is the Brier score, introduced by \textcite{Brier_MonthlyWeatherReview_1950}.
The Brier score measures the mean squared error between predicted probabilities and the actual binary outcomes.
For forecast--outcome pairs, the Brier score is defined as:

\[
  \mathrm{BS} = \operatorname{E}\!\left[(F - Y)^2\right]
\]

where $F$ is the random variable representing the predicted probability of the event, $Y \in \{0, 1\}$ is the random variable representing the binary outcome (1 if the event occurs, 0 otherwise), and $\operatorname{E}[\cdot]$ denotes the expectation over the joint distribution of $(F, Y)$.
The Brier score ranges from 0 to 1, with lower values indicating better predictive accuracy.

Decompositions of the Brier score are useful because they provide insights into different aspects of forecast performance.
By breaking down the Brier score into components, we can better understand the sources of forecast errors and identify areas for improvement.

In this note, we will adopt the nomenclature of \textcite{Mitchell_UniversityofExeter_2019}.

Multiple decompositions of the Brier score have been proposed in the literature~\cite{Sanders_JApplMeteorolClimatol_1963,Murphy_JApplMeteorolClimatol_1973,Yates_OrganizationalBehaviorandHumanPerformance_1982,Murphy_MonWeatherRev_1987}.

\subsection{Uncertainty, Resolution \& Reliability}

First, the Brier Score was decomposed into Sharpness \& Reliability components by \textcite{Sanders_JApplMeteorolClimatol_1963} by conditioning on the predicted probabilities $F$.

\begin{align*}
  \mathrm{BS} &= \operatorname{E}\!\left[\left(F - Y\right)^2\right] \\
  &= \underbrace{\operatorname{E}\!\left[\operatorname{Var}(Y \mid F)\right]}_\text{Sharpness} + \underbrace{\operatorname{E}\!\left[\left(F - \operatorname{E}[Y \mid F]\right)^2\right]}_\text{Reliability}
\end{align*}

The Sharpness component measures the average width of the predicted probabilities.
It is defined as the mean squared difference between the predicted probabilities and the climatological probabilities.
The climatological probabilities are the long-term average of the actual outcomes.
The Sharpness component is a measure of the informativeness of the predictions.

The Reliability component measures the calibration of the predicted probabilities.
It is defined as the mean squared difference between the predicted probabilities and the empirical frequencies of the actual outcomes.
The empirical frequencies are the observed proportions of the actual outcomes.
The Reliability component is a measure of the accuracy of the predictions.

Sharpness can be further decomposed into Resolution \& Uncertainty components using the law of total variance.

\begin{align*}
  \text{Sharpness} &= \operatorname{E}\!\left[\operatorname{Var}(Y \mid F)\right] \\
  &= \underbrace{\operatorname{Var}(Y)}_\text{Uncertainty} - \underbrace{\operatorname{Var}\!\left(\operatorname{E}[Y \mid F]\right)}_\text{Resolution}
\end{align*}

Uncertainty measures the variance of the climatological probabilities.
Resolution measures the conditional variance of the actual outcomes given the predicted probabilities.

Together, Uncertainty, Resolution and Reliability form the URR decomposition of the Brier score~\cite{Murphy_JApplMeteorolClimatol_1973,Mitchell_UniversityofExeter_2019}.

\subsection{Refinement, Discrimination, Correctness}

An alternative decomposition of the Brier score can be obtained by conditioning on the actual outcomes instead of the predicted probabilities.
First, by conditioning on the actual outcomes $Y$, we can decompose the Brier score into Excess \& Correctness components as

\begin{align*}
  \mathrm{BS} &= \operatorname{E}\!\left[\left(F - Y\right)^2\right] \\
  &= \underbrace{\operatorname{E}\!\left[\operatorname{Var}(F \mid Y)\right]}_\text{Excess} + \underbrace{\operatorname{E}\!\left[\left(\operatorname{E}[F \mid Y] - Y\right)^2\right]}_\text{Correctness}.
\end{align*}

Excess measures the average width of the predicted probabilities, conditional on the actual outcomes, and it can be further decomposed into Refinement \& Discrimination components using the law of total variance as

\begin{align*}
  \text{Excess} &= \operatorname{E}\!\left[\operatorname{Var}(F \mid Y)\right] \\
  &= \underbrace{\operatorname{Var}(F)}_\text{Refinement} - \underbrace{\operatorname{Var}\!\left(\operatorname{E}[F \mid Y]\right)}_\text{Discrimination}.
\end{align*}

This results in the Refinement, Discrimination, Correctness decomposition of the Brier score~\cite{Murphy_MonWeatherRev_1987,Mitchell_UniversityofExeter_2019}.
Refinement measures the variance of the probabilities.
Discrimination measures the conditional variance of the predicted probabilities given the actual outcomes.
Correctness measures the squared difference between the true outcomes and the predicted probabilities, conditional on the actual outcomes.

\subsection{Yates Covariance Decomposition}

Starting from the bias--variance decomposition of the Brier score,

\[
  \mathrm{BS} = \underbrace{\operatorname{Var}(F - Y)}_\text{variance term} + \underbrace{(\mu_F - \mu_Y)^2}_\text{squared bias term},
\]

where $\mu_F = \operatorname{E}[F]$ and $\mu_Y = \operatorname{E}[Y]$, one can further expand the variance term using $\operatorname{Var}(F-Y) = \operatorname{Var}(F) + \operatorname{Var}(Y) - 2\Cov(F,Y)$ to obtain

\begin{align*}
  \mathrm{BS} &= \operatorname{E}\!\left[(F - Y)^2\right] \\
  &= \operatorname{Var}(F - Y) + (\mu_F - \mu_Y)^2 \\
  &= \sigma_F^2 + \sigma_Y^2 - 2\sigma_{FY} + (\mu_F - \mu_Y)^2,
\end{align*}

where $\sigma_F^2 = \operatorname{Var}(F)$ is the variance of the predicted probabilities, $\sigma_Y^2 = \operatorname{Var}(Y)$ is the variance of the outcomes, $\sigma_{FY} = \Cov(F, Y)$ is the covariance between the predicted probabilities and the outcomes, and $(\mu_F - \mu_Y)^2$ is the squared difference between the mean predicted probability and the mean outcome, often referred to as calibration-in-the-large.
This is known as the Yates (covariance) decomposition of the Brier score~\cite{Yates_OrganizationalBehaviorandHumanPerformance_1982}.

The optimal forecast is the one that minimizes the Brier score.
However, it is not the case that minimizing each component individually minimizes the Brier score itself.
This can be seen in $\sigma_F^2$, where the optimal forecast necessarily does not minimize this term.
This was previously noted by Yates himself, who stated~\cite{Yates_OrganizationalBehaviorandHumanPerformance_1982} that
\begin{quote}
  \textit{
    the aim of the forecaster should be to minimize the variance of his or her forecasts, $S_f^2$.
    There is an obvious qualification on this advice, however.
    The only way $S_f^2$ can take on its absolute minimum possible value of zero is when the forecaster offers constant forecasts.
    This strategy would make the covariance term zero, too.
    So the proper objective of the forecaster should be to minimize $S_f^2$, \emph{given} that he or she exercises his or her fundamental forecasting abilities, as represented by $S_{fd}$.
  }
\end{quote}
We believe that the conventional presentation of the decomposition does not allow for an intuitive communication of why this is the case, and we propose a new decomposition arrangement that makes this more explicit.

\section{Alternative arrangement of the covariance decomposition}

\begin{theorem}[Alternative Yates Decomposition]\label{thm:main}
  The Yates covariance decomposition of the Brier score can be rearranged as
  \begin{align*}
    \mathrm{BS} &= \underbrace{(\sigma_F - \sigma_Y)^2}_{\text{variance mismatch}} \\
    &\quad+ \underbrace{2 (\sigma_F \sigma_Y - \sigma_{FY})}_{\text{covariance deficit}} \\
    &\quad+ \underbrace{(\mu_F - \mu_Y)^2}_{\text{calibration-in-the-large}}.
  \end{align*}
\end{theorem}

\begin{proof}
  Expanding the first two terms of the right-hand side:
  \begin{multline*}
    (\sigma_F - \sigma_Y)^2 + 2(\sigma_F \sigma_Y - \sigma_{FY}) \\
    = \sigma_F^2 - 2\sigma_F \sigma_Y + \sigma_Y^2 + 2\sigma_F \sigma_Y - 2\sigma_{FY} \\
    = \sigma_F^2 + \sigma_Y^2 - 2\sigma_{FY},
  \end{multline*}
  which, together with the calibration-in-the-large term $(\mu_F - \mu_Y)^2$, recovers the original Yates decomposition.
\end{proof}

\begin{corollary}[Non-negativity]\label{cor:nonneg}
  All three terms in \cref{thm:main} are non-negative.
\end{corollary}

\begin{proof}
  The first and third terms are squares and hence non-negative.
  For the second term, the Cauchy--Schwarz inequality gives $|\sigma_{FY}| \leq \sigma_F \sigma_Y$, so $2(\sigma_F \sigma_Y - \sigma_{FY}) \geq 0$.
\end{proof}

\begin{corollary}[Optimality conditions]\label{cor:optimal}
  $\mathrm{BS} = 0$ if and only if $F = Y$ almost surely.
  Equivalently, each of the following conditions must hold simultaneously:
  \begin{enumerate}
    \item $\sigma_F = \sigma_Y$ \emph{(variance matching)},
    \item $\sigma_{FY} = \sigma_F \sigma_Y$ \emph{(perfect positive correlation)},
    \item $\mu_F = \mu_Y$ \emph{(no bias)}.
  \end{enumerate}
\end{corollary}

\begin{proof}
  By \cref{cor:nonneg}, all three terms are non-negative, so $\mathrm{BS} = 0$ if and only if each term vanishes individually.
  Term~1 vanishes iff $\sigma_F = \sigma_Y$;
  term~2 vanishes iff $\sigma_{FY} = \sigma_F \sigma_Y$ (i.e., $\rho_{FY} = 1$ when the standard deviations are non-zero);
  term~3 vanishes iff $\mu_F = \mu_Y$.
  These three conditions together are equivalent to $F = Y$ almost surely.
\end{proof}

\begin{remark}
  When $\sigma_F$ and $\sigma_Y$ are both non-zero, the covariance deficit term can be rewritten in terms of the Pearson correlation $\rho_{FY} = \sigma_{FY} / (\sigma_F \sigma_Y) = \Corr(F, Y)$ as
  \[
    2 (\sigma_F \sigma_Y - \sigma_{FY}) = 2 \sigma_F \sigma_Y (1 - \rho_{FY}),
  \]
  which is minimized when $\rho_{FY} = 1$.
  This makes explicit the point noted by \textcite{Yates_OrganizationalBehaviorandHumanPerformance_1982}: a forecaster should not simply minimize $\sigma_F^2$, but rather match the variability of forecasts to that of outcomes while maintaining maximal correlation.
\end{remark}

Under the notation of \textcite{Mitchell_UniversityofExeter_2019}, this means that under the optimality conditions, the Refinement term $\sigma_F^2$ is equal to the Uncertainty term $\sigma_Y^2$, a well known result.

\section{Conclusion}

We have shown that a simple algebraic rearrangement of the Yates covariance decomposition yields three independently non-negative terms, each with a clear interpretation: variance mismatch, covariance deficit, and calibration-in-the-large.
The non-negativity of the covariance deficit follows directly from the Cauchy--Schwarz inequality.
This decomposition makes the optimality conditions for the Brier score transparent: a perfect forecast must simultaneously match the variance of outcomes, achieve perfect positive correlation ($\rho_{FY} = 1$), and exhibit no bias ($\mu_F = \mu_Y$).
In particular, the rearrangement resolves the interpretive difficulty noted by \textcite{Yates_OrganizationalBehaviorandHumanPerformance_1982} regarding the role of forecast variance $\sigma_F^2$ in the original decomposition, since the variance mismatch term $(\sigma_F - \sigma_Y)^2$ is clearly minimized not by minimizing the variability of forecasts but by matching it to that of outcomes.

\printbibliography

\end{document}